\documentclass[prl,floatfix,amsmath,superscriptaddress,twocolumn]{revtex4}
\usepackage{amssymb}
\usepackage{graphicx}
\usepackage{graphics}
\usepackage{amsmath}
\usepackage{amsthm}
\usepackage{color}
\definecolor{BLACK}{gray}{0}
 \definecolor{WHITE}{gray}{1}
 \definecolor{RED}{rgb}{1,0,0}
 \definecolor{GREEN}{rgb}{0,1,0}
 \definecolor{BLUE}{rgb}{0,0,1}
 \definecolor{CYAN}{cmyk}{1,0,0,0}
 \definecolor{MAGENTA}{cmyk}{0,1,0,0}
 \definecolor{YELLOW}{cmyk}{0,0,1,0}

\def\h{{\bf h}}
\def\g{{\bf g}}

\newcommand{\bea}{\begin{eqnarray}}
\newcommand{\eea}{\end{eqnarray}}

\def\bi{\begin{itemize}}
\def\ei{\end{itemize}}
\def\bc{\begin{center}}
\def\ec{\end{center}}

\def\C{\hbox{$\mit I$\kern-.7em$\mit C$}}
\def\R{\hbox{$\mit I$\kern-.6em$\mit R$}}
\def\N{\hbox{$\mit I$\kern-.6em$\mit N$}}

\def\ket#1{|#1\rangle}
\newcommand{\one}{\mbox{$1 \hspace{-1.0mm}  {\bf l}$}}
\def\tr{\mathrm{tr}}
\def\ket#1{\left| #1\right>}

\newtheorem{theorem}{Theorem}
\newtheorem{corollary}[theorem]{Corollary}
\newtheorem{lemma}[theorem]{Lemma}

\begin{document}

\author{}

\author{C. Spee}
\affiliation{Institute for Theoretical Physics, University of
Innsbruck, Innsbruck, Austria}
\author{J.I. de Vicente}
\affiliation{Departamento de Matem\'aticas, Universidad Carlos III de
Madrid, Legan\'es (Madrid), Spain}
\author{D. Sauerwein}
\affiliation{Institute for Theoretical Physics, University of
Innsbruck, Innsbruck, Austria}
\author{B. Kraus}
\affiliation{Institute for Theoretical Physics, University of
Innsbruck, Innsbruck, Austria}
\title{Entangled pure state transformations via local operations assisted by finitely many rounds of classical communication}

\begin{abstract}
We consider generic pure $n$--qubit states and a general class of pure states of arbitrary dimensions and arbitrarily many subsystems. We characterize those states which can be reached from some other state via Local Operations assisted by finitely many rounds of Classical
Communication ($LOCC_{\N}$). For $n$ qubits with $n>3$ we show that this set of states is of measure zero, which implies that the maximally entangled set is generically of full measure if restricted to the practical scenario of $LOCC_{\N}$. Moreover, we identify a class of states for which any $LOCC_{\N}$ protocol can be realized via a concatenation of deterministic steps. We show, however, that in general there exist state transformations which require a probabilistic step within the protocol, which highlights the difference between bipartite and multipartite LOCC.

\end{abstract}
\maketitle

Multipartite entanglement plays a crucial role in many fields of physics \cite{reviews}. This is in particular so if all the correlations among the constituent systems result from entanglement, which is the case for pure states. The existence of these nonclassical correlations both in the bipartite and multipartite case have been pivotal in the development of quantum information theory. In this context, there are many applications of pure multipartite entanglement including quantum computation \cite{RaBr01}, metrology \cite{reviewmet}, and quantum communication protocols
\cite{SecretSh,reviews}. Furthermore, the entanglement properties of multipartite states have been proven successful in the study of condensed matter physics like in e.g.\ the identification of different phases \cite{AmFa08} and the development of numerical methods \cite{orus}. A deep understanding of entanglement is central to all these investigations and this has led to the development of entanglement theory, which aims at providing a solid framework for its characterization, quantification and manipulation.

As entanglement theory is a resource theory, where the free operations are those which can be realized via Local Operations assisted by Classical Communication (LOCC), the investigation of the latter is essential in this theory. This provides all possible protocols for entanglement manipulation for spatially separated parties and induces an operationally meaningful ordering in the set of entangled states, which allows to quantify
and qualify entanglement. Bipartite pure state entanglement is well understood due to the fact that all LOCC transformations among bipartite pure states can be easily characterized \cite{nielsen}. There, transformations under a larger class of operations, the so-called separable operations (SEP), can always be realized via LOCC \cite{gheorghiu}. Multipartite LOCC is far from being that simple.
In fact, it has been shown that infinitely many rounds of communication might be necessary in certain scenarios involving ensembles of states \cite{chit11}.
Note, however, that to date there exists no example where infinitely
many rounds of communication are required for pure state transformations. Further aggravating the matter, separable pure state transformations have been identified which cannot be realized via LOCC,
even if infinitely many rounds are utilized \cite{HeSp15}. Hence, the investigation of the mathematically much more manageable separable transformations leads to necessary, but not sufficient conditions for the existence of a transformation among pure multipartite states via LOCC. Other approaches in gaining insight into the complicated structure of multipartite entanglement are based on Local Unitary (LU) transformations \cite{Kr10}, which do not alter the entanglement contained in the state, and Stochastic LOCC (SLOCC) transformations \cite{slocc}. Both of them define an equivalence relation, namely two $n$--partite states, $\ket{\Psi}$ and $\ket{\Phi}$ are LU (SLOCC)--equivalent if there exist unitary (invertible) matrices $A_i$ for $i\in \{1,\ldots ,n\}$ such that $\ket{\Psi}\propto A_1\otimes \ldots \otimes A_n\ket{\Phi}$, respectively. However, these notions cannot be utilized to establish any ordering among the entanglement contained in the states as LOCC does. Hence, despite the fact that the structure of LOCC maps is mathematically very subtle \cite{Donald,chitambar2,chitambar3,cohen}, the understanding of possible transformations under LOCC is necessary in order to clarify the usefulness of different
states and to quantify entanglement, which can be done by any quantity which does not increase under LOCC \cite{reviews}.

In \cite{dVSp13,SpdV16,HeSp15} we generalized the notion of maximal entanglement to the multipartite case by identifying the minimal set of states, the Maximally Entangled Set (MES), which suffices to reach any other state via LOCC. Moreover, the LOCC transformations
among three-, four-qubit and three-qutrit states have been investigated in \cite{HeSp15,dVSp13,SpdV16,turgut,SaSc15}. Here, we consider the more realistic scenario of LOCC protocols consisting of finitely many rounds of communication,
i.e. $LOCC_{\N}$. We investigate such transformations among pure truly $n$--partite entangled states, $\ket{\Psi}, \ket{\Phi} \in \C^{d_1}\otimes \ldots \otimes \C^{d_n}$, with $d_i$ denoting the local dimension of
subsystem $i$, i.e.\ the rank of the corresponding reduced state. We are interested in the case where the final and initial state are in the same SLOCC class \cite{slocc}. That is, we do neither investigate for instance transformations from entangled 4-qubit to entangled 3-qubit states where the 4--th qubit factorizes, nor transformations where the local dimensions, $d_i$, differ between the initial and the final state. Moreover, as we are only interested in non--trivial transformations, i.e. not LU transformations, we often refer by a state to a LU--equivalence class. Every LOCC protocol can then be described as follows. One party applies locally a measurement on his system and sends the information about the measurement outcome to the other parties, who then apply, depending on this information, LUs to their systems. Such rounds are concatenated until the transformation from the initial to the final states is accomplished deterministically. It is the dependency of the measurement on all the previous measurement outcomes
which makes LOCC so cumbersome to handle, even if only finitely many rounds are considered (see e.g.\cite{Donald,chitambar2,chitambar3,cohen}).

We investigate many SLOCC classes of states of arbitrary number of parties and local dimensions. In particular, for $n$-qubit states ($n>3$) and 3-qutrit states their union constitutes a generic set of states, i.e.\ of full measure. That is, there, our results apply to all states but a subset of measure zero. Despite the aforementioned difficulties, we present here as our first main result a succinct
characterization of all states $\ket{\Phi}$ in these classes which are reachable via $LOCC_{\N}$, i.e. for which there exists a state that can be
transformed (nontrivally, i.e. not with LUs) via some $LOCC_{\N}$ protocol into $\ket{\Phi}$ \footnote{As mentioned before we consider the initial and final states to belong to the same SLOCC class.}. Moreover, this set is of measure zero in the corresponding SLOCC class, which allows us to show that $n$--qubit states with $n>3$ are almost never reachable. The reason why such a general result can be derived is that the conditions for a state to be reachable are very stringent, which implies that only very particular states can be reached.
In order to explain the other results presented here, let us note that all LOCC transformations among pure states (including infinitely many rounds) studied so far can be realized via a particularly simple protocol \cite{HeSp15,dVSp13,SpdV16,turgut,SaSc15}. There, in each round the state is transformed {\it deterministically} into an intermediate (or the final) pure state. That is, for any measurement outcome the system is in a pure state and all these states are LU--equivalent. We call these protocols in the following all--deterministic. As these results hold for various numbers of subsystems and different local dimensions,
one might wonder whether every LOCC protocol can be divided into deterministic steps (as is also the case in the bipartite setting \cite{LoPopescu}).
We prove here, however, as our second main result, that this is not the case. In particular, we present an example of a pure state transformation where a probabilistic step is required (see Fig. 1). In contrast to that,
we identify classes of states for which indeed any protocol in $LOCC_{\N}$ can be divided into deterministic steps, which makes it particularly easy to analyze them. Note that these results clearly show the difference
between multipartite and bipartite LOCC.

The outline of the remainder of the paper is the following. After presenting our notation we define the SLOCC classes which are considered here. We characterize all states (in those SLOCC classes) which can be reached via $LOCC_{\N}$ and show that this set of states is of measure zero for $n$--qubit states. Next, we investigate which states are convertible, i.e. can be transformed into another state via $LOCC_{\N}$. This result can be used to characterize all--deterministic $LOCC_{\N}$ protocols, to which any previously known LOCC protocol belongs to. After that we show, however, that not any $LOCC_{\N}$ protocol is of this simple form by presenting a $LOCC_{\N}$ protocol which is not realizable via an all--deterministic transformation. We then briefly discuss that an interesting, but aggravating,
phenomenon can occur, namely that one party can unlock or lock the power of the other parties. That is, one party can enable or prevent the other parties to perform a deterministic step. Considering
instances where this phenomenon cannot occur, we identify a class of states for which any $LOCC_{\N}$ transformation can be realized via an all--deterministic $LOCC_{\N}$ protocol.

We denote throughout this paper by $\ket{\Psi_s}$ a $n$--partite state whose local stabilizer, $S_{\Psi_s}$, consists of finitely many LUs \cite{Gour}. That is, there exist only finitely many operators $S=S^{(1)} \otimes S^{(2)}\otimes \ldots \otimes S^{(n)}$, with $S\ket{\Psi_s}=\ket{\Psi_s}$ \footnote{Note that due to the small number of parameters the existence of a nontrivial local symmetry is already very stringent.}. Moreover, these operators are all unitary \footnote{As the stabilizer is a group, the fact that there exists $\ket{\Psi}$ such that $S_{\Psi}$ is finite implies that there exists another state in the same SLOCC class whose stabilizer is finite and unitary (see e.g. \cite{Gour}).}. Here and in the following, the superscript $(i)$ refers to the systems on which the
operator is acting on. It has been shown in \cite{bookWallach} that the stabilizer of a generic $n$-qubit state ($n>3$) is finite. Hence, such states can be written as  $g\ket{\Psi_s}$, with $g=\otimes_{i=1}^n g_i \in G\equiv GL(2)^{\otimes n}$ and $\ket{\Psi_s}$ as above. For qudit states some of the SLOCC classes also possess a representative which has only finitely many local symmetries. Moreover, as for instance in the case of three qutrits, these SLOCC classes can be generic too \cite{BrLu04}. All our results apply to any SLOCC class which can be represented by a state $\ket{\Psi_s}$ with $S_{\Psi_s}$ finite. The reason why the symmetries of $\ket{\Psi_s}$ are so important in these investigations becomes clear by noting that any local operator which maps a state $g\ket{\Psi_s}$ into a
state $h\ket{\Psi_s}$ must be of the form $h_1 S^{(1)} g_1^{-1}\otimes \ldots \otimes h_n S^{(n)} g_n^{-1}$, with $S\in S_{\Psi_s}$. Hence, deciding whether a transformation is possible (deterministically) depends very crucially on the properties of the stabilizer. In the following we choose $G_i=g_i^\dagger g_i$ such that $\tr(G_i)=1$ for any $i$ and similarly for $H_i=h_i^\dagger h_i$.

Let us now show which states, $\ket{\Phi}\propto h \ket{\Psi_s}$, are reachable via $LOCC_{\N}$ (from a state $\ket{\Psi}\propto g\ket{\Psi_s}$).

\begin{theorem} \label{theorem_1}
A state $\ket{\Phi}\propto h \ket{\Psi_s}$ is reachable, iff there exists $S \in S_{\Psi_s}$ such that the following conditions hold up to permutations of the particles:
\bi \item[i)] For any $i\geq 2$ $[H_i, S^{(i)}]=0$ and
\item[ii)] $[H_1, S^{(1)}]\neq 0$.\ei
\end{theorem}

\begin{proof}
Let us first show that the conditions in the theorem are necessary and then construct the state $\ket{\Psi}$ which can be transformed to $\ket{\Phi}$.

As the protocol is finite there has to exist a last step of the protocol. At this step, there must exist a deterministic transformation from some state $\ket{\chi}$, which is obtained in one branch of the LOCC protocol, to $\ket{\Phi}$. As these two states need to be in the same SLOCC class, we write $\ket{\chi}\propto g \ket{\Psi_s}$ for some $g\in G$ \footnote{Notice that the local dimension cannot increase in any branch of an LOCC protocol. Also $\ket{\chi}$ can neither have smaller local dimension as otherwise it could not be transformed to $\ket{\Phi}$.}. W.l.o.g. we assume that party $1$ applies, at this step, a measurement, which we describe by the operators $\{A_i\}$, whereas all the other parties only apply LUs. Note that as the protocol is nontrivial, there must exist at least two outcomes, which are not related to each other by a unitary, i.e. $A_2^\dagger A_2 \not\propto
A_1^\dagger A_1$. Considering these two outcomes it must hold that $(A_1\otimes \one) g\ket{\Psi_s}=r_1 (\one \otimes_{i=2}^n U_i )\ket{\Phi}$ and $(A_2\otimes \one)g\ket{\Psi_s}=r_2 (\one \otimes_{i=2}^n V_i) \ket{\Phi}$
for some local unitaries $U_i, V_i$. The real numbers $r_1, r_2$ can be chosen strictly positive as $A_i\otimes \one \ket{\chi}=0$ implies, as $\ket{\chi}$ is in the same SLOCC class as $\ket{\Phi}$ and therefore the reduced states have full rank, that $A_i=0$. Using the symmetries of $\ket{\Psi_s}$, the equations above are equivalent to
\bea h^{-1}  (A_1 \otimes_{i=2}^n U_i^\dagger) g&=&r_1 S_1 \\
h^{-1}  (A_2 \otimes_{i=2}^n V_i^\dagger ) g&=&r_2 S_2,\eea where $S_1,S_2 \in S_{\Psi_s}$. Hence, we have \bea
A_1&=&r^{(1)}_1 h_1S_1^{(1)}g_1^{-1}, A_2=r^{(1)}_2 h_1S_2^{(1)}g_1^{-1}\\
g_i&=&r^{(i)}_1 U_i h_i S_1^{(i)}=r^{(i)}_2 V_i h_iS_2^{(i)}, \quad \forall i>1.\eea Here, $r_j=\prod_i r^{(i)}_j$, for $j=1,2$. Considering now the last equations for $g_i^\dagger g_i$ and using that $h_i$ is
invertible, one easily finds that $r^{(i)}_1=r^{(i)}_2$ $\forall i>1$ and therefore that the conditions (i) in Theorem \ref{theorem_1} are necessary for $S = S_1 S_2^\dagger$. Moreover, using that
$A_1^\dagger A_1 \not\propto A_2^\dagger A_2$ we find that condition (ii) is necessary for $S = S_1 S_2^\dagger$.

The construction of the state  $\ket{\Psi}\propto g \ket{\Psi_s}$ and the corresponding LOCC protocol which transforms $\ket{\Psi}$ into $\ket{\Phi}$ is now very simple. Choosing for
$i>1$ $G_i = H_i=(S^{(i)})^\dagger H_i S^{(i)}$, i.e. choosing $g_i=V_i h_i =W_i h_i S^{(i)}$, for some unitaries $V_i,W_i$, which have to exist as condition (i) is equivalent to the condition that
$h_i S^{(i)} (h_i)^{-1}$ is
unitary, and $G_1=p  H_1 + (1-p) (S^{(1)})^\dagger H_1 S^{(1)}$, for some $0<p<1$ allows to reach the state with the following LOCC protocol. Party $1$ measures the POVM consisting of the measurement operators
$ \sqrt{p}h_1 g_1^{-1},
\sqrt{1-p}h_1 S^{(1)} g_1^{-1}$. Depending on the outcome of this measurement, all the other parties $i$ apply either $V^\dagger_i$ or $W^\dagger_i$, respectively.

\end{proof}

Hence, once the symmetries of $\ket{\Psi_s}$ are known, it is very easy to decide whether a state is reachable via $LOCC_{\N}$ or not. For instance, consider $\ket{\Psi_s}$ with symmetries $\sigma_i^{\otimes n}$,
where $\sigma_i$ denotes here and in the following the Pauli operators. Then, due to Theorem \ref{theorem_1} it is straightforward to see that the state $h_1\otimes \one \ket{\Psi_s}$ is reachable for arbitrary $h_1$. However, the state
$h_1\otimes h_2\otimes \one \ket{\Psi_s}$ is not, if neither $H_1$ nor $H_2$ commutes with $\sigma_i$ for some $i$. As mentioned before, the considered SLOCC classes are generic for $n>3$ qubit states \cite{bookWallach}. Hence, Theorem \ref{theorem_1} characterizes (almost) all reachable states there. Due to that and the fact that for almost all states $h \ket{\Psi_s}$, the operator $h$ does not obey the commutation relations stated in Theorem \ref{theorem_1}, we obtain the following corollary, which we prove in \cite{SaDeV162}.

\begin{corollary}\label{corollary2}
The set of $n$--qubit states ($n>3$) which are reachable via a $LOCC_{\N}$ protocol is of measure zero. \end{corollary}

Note that this result applies also to all multipartite states of higher local dimensions, as long as the considered SLOCC classes are generic. This means that the MES (under $LOCC_{\N}$) has full measure in this case. Let us now investigate which states are nontrivially transformable to another state, i.e., are convertible. We call a state convertible via $LOCC_j$, if it can be converted by a single round of LOCC, where the nontrivial
measurement is performed by party $j$ and LUs are applied by the other parties. Considering w.l.o.g. that party 1 applies the measurement and using similar tools as in the proof of Theorem \ref{theorem_1}, one can easily prove
the following lemma \cite{SaDeV162}.
\begin{lemma} \label{lemmaconv}
A state $\ket{\Psi}\propto g\ket{\Psi_s}$ is convertible via $LOCC_1$ iff there exist $m$ symmetries $S_i\in S_{\Psi_s}$, with $m> 1$ and $H\in {\cal B} ({\cal H}_1)$, $H >0$ and $p_i> 0$ with $\sum_{i=1}^m p_i=1$ such that the following conditions hold
\bi \item[(i)] $[G_k,S_i^{(k)}]=0$ $\forall k>1$ and $\forall i \in \{1,\ldots,m\}$
\item[(ii)] $G_1=\sum_{i=1}^m p_i (S_i^{(1)})^\dagger H S_i^{(1)}$ and $H \neq S^{(1)} G_1 (S^{(1)})^\dagger$ for any $S \in S_{\Psi_s}$ fulfilling (i).
\ei
\end{lemma}

Note that the first party can apply measurement operators $\{A_i= \sqrt{p_i} h S_i g_1^{-1}\}_{i=1}^m$ with probabilities $p_i=\tr(g_1^\dagger A_i^\dagger A_ig_1)$,
where $H=h^\dagger h$ is such that the conditions in (ii) are satisfied. Depending on the measurement outcome $i$, the other parties apply the LUs $U_i^{(k)}$ defined by
$U_i^{(k)} g_k=  g_k S_i^{(k)}$ $\forall k>1$ to obtain $h\otimes g_2 \otimes \ldots \otimes g_n\ket{\Psi_s}$. These unitaries exist due to condition (i).

Using Lemma \ref{lemmaconv} it is now straightforward to characterize all possible all--deterministic $LOCC_{\N}$ transformations, which can be viewed as a generalization of the bipartite transformations. All known LOCC transformations
among pure states are exactly of this kind. Moreover, we note that the set of states in the
MES (for $LOCC_{\N}$) which are convertible via  all--deterministic $LOCC_{\N}$ can be then characterized by simple conditions \cite{SaDeV162}. However, as we show in the following by constructing an explicit example, it turns out that all--deterministic transformations are not the most general ones. That is, certain transformations can only be accomplished by using an intermediate probabilistic step (see Fig.\ 1). This result shows that the involved structure of LOCC maps can be exploited to achieve pure-state transformations and exposes once again the difficulty of a general characterization.

In order to provide the aforementioned example we consider the SLOCC class given by the $L$-state of four qubits \cite{GoWa,SpdV16},
\begin{equation}\label{lstate}
|L\rangle=\frac{1}{\sqrt{3}}(|\phi^-\rangle|\phi^-\rangle+e^{i\pi/3}|\phi^+\rangle|\phi^+\rangle
+e^{i2\pi/3}|\psi^+\rangle|\psi^+\rangle),
\end{equation}
where $\ket{\phi^\pm}=(\ket{00}\pm \ket{11})/\sqrt{2}$ and $\ket{\psi^\pm}=(\ket{01}\pm \ket{10})/\sqrt{2}$. The symmetries of this state are given by $S_L=\{\{\one, U, U^2 \}\times \{\sigma_i\}_{i=0}^3\}^{\otimes4}$, where $U=\sqrt{i \sigma_y}\sqrt{i \sigma_x}$ \cite{SpdV16}. We will consider states of the form $g_1\otimes g_2\otimes\one\otimes\one|L\rangle$, which we denote in the following by $\{\g_1,\g_2\}$, where $\g_i$ denotes the Bloch vector of $G_i$ with $g_i=\sqrt{ \one/2 + \g_i \cdot \vec{\sigma}}$. The above referred example is given by the transformation $\{\g_1,\g_2\}\to\{\h_1,\h_2\}$, where $\{\g_1,\g_2\}=\{(x,x,2x),(x,-x,0)\}$ and $\{\h_1,\h_2\}=\{2(x,x,x),(x,x,-2x)\}$ with $x>0$ but small enough so that the corresponding operators are well defined. Notice that, since for $i=1,2$
$[G_i,S^{(i)}]\neq0$ $\forall S \in S_L$ $(S\neq\one)$, Lemma \ref{lemmaconv} guarantees that the initial state cannot be converted by an LOCC$_j$ protocol $\forall j$ and, hence, any deterministic
transformation starting from this state necessarily requires intermediate non-deterministic steps. The corresponding protocol has two steps. First, party 1 implements a two-outcome POVM that leads to the intermediate states $h_1\otimes g_2\otimes\one\otimes\one|L\rangle$ and $h_1 \sigma_3 \otimes g_2\otimes\one\otimes\one|L\rangle$ by measuring $M_1=\sqrt{3/4}h_1g_1^{-1}$ and $M_2=\sqrt{1/4}h_1\sigma_3g_1^{-1}$ (which fulfills $\sum_iM_i^\dag M_i=\one$). Since $[H_1,U]=0$, both intermediate states fulfill the premises of Lemma \ref{lemmaconv} so that they can be now transformed by LOCC$_2$ into the desired state. For this, in the first branch of the protocol party 2 measures $M'_1=\sqrt{1/3}h_2g_2^{-1}$ and $M'_2=\sqrt{2/3}h_2U^2g_2^{-1}$ (which is again a valid measurement), leading to the states $h_1\otimes h_2\otimes\one\otimes\one|L\rangle$ and $h_1\otimes h_2 U^2\otimes\one\otimes\one|L\rangle$. For the second outcome, parties 1, 3 and 4 additionally apply the unitary $U^2$, obtaining then $\{\h_1,\h_2\}$ as well since $[h_1,U^2]=0$ and $U^2\in S_L$. Analogously, in the second branch of the protocol party 2 measures $M''_1=\sqrt{1/3}h_2\sigma_3g_2^{-1}$ and $M''_2=\sqrt{2/3}h_2U\sigma_3g_2^{-1}$. In case of the first outcome, parties 3 and 4 apply the unitary $\sigma_3$ and, in case of the second, party 1 applies $U$ and 3 and 4 apply $U\sigma_3$ obtaining again the state $\{\h_1,\h_2\}$. In \cite{SaDeV162} we analyze further how these constructions arise and we generalize them.

\begin{figure}
   \centering
   \includegraphics[width=0.4\textwidth]{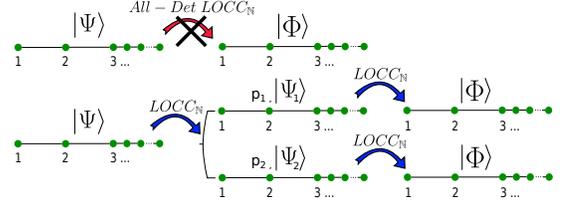}
   \caption{The transformation from the state $\ket{\Psi}$ to $\ket{\Phi}$ is impossible with a all--deterministic $LOCC_{\N}$. However, it becomes possible if party 1 performs a non--deterministic step, transforming $\ket{\Psi}$ with probability $p_1$ into $\ket{\Psi_1}$ and with probability $p_2$ into $\ket{\Psi_2}$. Both states can then be transformed deterministically into the final state $\ket{\Phi}$. Note that this example is in clear contrast to bipartite state transformation, where any transformation can be performed with a all--deterministic $LOCC_{\N}$.}
   \label{VaVs2}
   \end{figure}

It is worth mentioning that multipartite $LOCC_{\N}$ manipulation allows for an interesting phenomenon that we name locking or unlocking the power of other parties: it can be that the action of some party prevents or allows the others to perform deterministic nontrivial transformations. In \cite{SaDeV162} we provide examples of this feature and analyze general conditions on $S_{\Psi_s}$ that are necessary for unlocking to be possible. Also, imposing further conditions on $S_{\Psi_s}$ allows to find SLOCC classes where any possible $LOCC_{\N}$ transformation can be realized via an all--deterministic transformation. An instance is the case $S_{\Psi_s}=\{\sigma_i^{\otimes n}\}_i$ (see \cite{SaDeV162} for the proof). Moreover, for these classes SEP (and, hence, also infinitely many round LOCC) protocols are all--deterministic.

We have investigated $LOCC_{\N}$ transformations among pure states in certain SLOCC classes of arbitrary dimension and system sizes. We characterized all reachable states in Theorem \ref{theorem_1}. Moreover, we provided examples of SLOCC classes where even any SEP transformation is all--deterministic $LOCC_{\N}$. That is, in each step of a protocol a deterministic LOCC transformation is performed. All these transformations
can then be characterized
using Lemma \ref{lemmaconv}. However, we showed that there exist transformations among pure states that require more elaborate LOCC protocols which include non-deterministic intermediate steps. This fact prevents
the characterization of pure-state $LOCC_{\N}$ transformations via Theorem \ref{theorem_1} and  Lemma \ref{lemmaconv} from being complete. In summary, putting these results together with previous investigations on LOCC \cite{dVSp13,SpdV16,HeSp15} the following picture emerges. While for bipartite pure state transformations we have that all--deterministic $LOCC_{\N} = LOCC_{\N}=LOCC=SEP$, in the multipartite case we showed that all--deterministic $LOCC_{\N}\subsetneq LOCC_{\N}$ and $LOCC \subsetneq SEP$. It remains open whether $LOCC_{\N}=LOCC$ for pure states, which would be an interesting topic for future research. Our results also show that exact LOCC transformations among pure states are rarely possible (cf. Corollary \ref{corollary2}).  This suggests further work in order to develop new tools to investigate approximate transformations.

\begin{acknowledgments}
We thank Gilad Gour and Nolan Wallach for pointing out and explaining the proof presented in Ref. \cite{bookWallach} to us.
After completing this manuscript it was proven that generic n-qubit-states ($n\geq 5$) have actually a trivial stabilizer \cite{GoKr16}. This gives an immediate alternative proof of our Corollary 2 in this case, which is moreover extendable to infinite-round protocols. The research of CS, DS and BK was funded by the Austrian Science Fund (FWF): Y535-N16. DS and BK thank in addition the DK-ALM: W1259-N27. The research of JIdV was funded by the Spanish MINECO through grants MTM2014-54692-P and MTM2014-54240-P and by the Comunidad de Madrid through grant QUITEMAD+CM S2013/ICE-2801.
\end{acknowledgments}

\newpage

\end{document}